\documentclass[12pt,twoside]{article}
\usepackage{amsmath,amsfonts,amssymb}
\usepackage{graphicx}
\usepackage{eucal}
\usepackage{mathrsfs}
\usepackage{theorem}
\usepackage{pifont}
\usepackage {epsfig}
\usepackage {graphicx}
\newcommand{\Real}{\mathbb{R}}

\newcommand{\todo}[1]{{\sffamily To do:}}

\newtheorem{theorem}{Theorem}

\newtheorem {lemma}{Lemma}
\newenvironment{proof}{{\flushleft \emph{Proof}:}}{\ding{110}}

\title{Existence of periodic solutions for the periodically forced SIR model}
\author{Guy Katriel\\Department of Mathematics, Ort Braude College\\ Karmiel, Israel \footnote{email: katriel@braude.ac.il}}

\date{}
\begin{document}

\maketitle

\begin{abstract}We prove that the seasonally-forced SIR model with a $T$-periodic forcing has a periodic solution with period $T$ whenever the basic reproductive number ${\cal{R}}_0>1$.
The proof uses the Leray-Schauder degree theory.
We also describe some numerical results in which we compute the $T$-periodic solution, where in order to obtain the $T$-periodic solution when the behavior of the system
is subharmonic or chaotic, we use a Galerkin scheme.
\end{abstract}


\section{Introduction}
The periodically forced SIR model
\begin{equation}\label{S}S'=\mu(1-S) -\beta(t) SI,\end{equation}
\begin{equation}\label{I}I'=\beta(t) SI-(\gamma+\mu) I,\end{equation}
\begin{equation}\label{R}R'=\gamma I -\mu R,\end{equation}
and variants of it, are extensively used to model seasonally recurrent diseases \cite{aron,bolker,dietz,earn,grassly,greenman,jodar,katriel,keeling,keeling1,kuznetsov,london,olinky}.
Here $S,I,R$ are the fractions of the population which are susceptible, infective, and recovered, $\mu$ denotes the birth and death rate, $\gamma$ the recovery rate, and $\beta(t)$, which we assume is a positive continuous $T$-periodic function, is the seasonally-dependent transmission rate (so that $T$ is the yearly period).

When simulating this model numerically (see section \ref{numerical}), it is observed that:

(i) If ${\cal{R}}_0\leq 1$, where
$${\cal{R}}_0=\frac{\bar{\beta}}{\gamma+\mu},$$
$$\bar{\beta}=\frac{1}{T}\int_0^T \beta(t)dt,$$
then all solutions tend to the disease free equilibrium $S_0=1, I_0=0, R_0=0$. This fact can be rigorously proved, see \cite{ma}.

(ii) If ${\cal{R}}_0> 1$, then depending on the values of the parameters, one observes convergence to $T$-periodic orbits, or to $nT$-periodic orbits with $n>1$
(subharmonics), or chaotic behavior.

A fundamental question, that is addressed here, is the {\it{existence}} of a $T$-periodic solution of the system. We demand, of course that the components $S(t),I(t),R(t)$
of this solution will be positive.  Obviously when ${\cal{R}}_0\leq 1$, a positive periodic
solution cannot exist, because such a solution would not converge to the disease-free equilibrium. We will prove, however, that
\begin{theorem}\label{main} Whenever ${\cal{R}}_0>1$,
there exists at least one $T$-periodic solution $(S(t),I(t),R(t))$ of (\ref{S})-(\ref{R}), all of whose components are positive.
\end{theorem}

Thus, when $T$-periodic behavior is not observed in simulations, this is not due to the fact that such a
solution does not exist, but rather to the fact that all $T$-period solutions are unstable.

Despite the fact that the existence of a $T$-periodic solution of the $T$-periodically forced SIR model is a fundamental issue, the only paper
in the literature of which we are aware to have dealt with this question is the recent paper of J\'odar, Villanueva and Arenas \cite{jodar}. They treated a more
general system then we do here (including loss of immunity and allowing other coefficients besides $\beta(t)$ to be $T$-periodic). Restricting their
existence result to the case of the SIR model (\ref{S})-(\ref{R}), they proved, using Mawhin's continuation theorem, that a $T$-periodic solution exists whenever the condition
\begin{equation}\label{jodar}\min_{t\in\Real}\beta(t)>\gamma+\mu\end{equation}
holds.
Note that the condition (\ref{jodar}) implies that $\bar{\beta}>\gamma+\mu$, that is ${\cal{R}}_0>1$, but it is a stronger condition. Theorem \ref{main}
uses only the condition ${\cal{R}}_0>1$, so that together with the fact noted above, that when ${\cal{R}}_0\leq 1$ a $T$-periodic solution does {\it{not}}
exist, we have that ${\cal{R}}_0>1$ is a {\it{necessary and sufficient}} condition for the existence of a $T$-periodic solution with positive components.

Our technique for proving Theorem \ref{main} relies on nonlinear functional analysis, for which we refer to the
textbooks \cite{brown,zeidler}. Reformulating the problem as one of solving an equation
in an infinite dimensional space of periodic functions, we define
a homotopy between the periodically forced problem and the autonomous problem in which $\beta(t)$ is replaced by the mean $\bar{\beta}$. The
autonomous problem has an endemic equilibrium, which is a trivial periodic solution.
We then employ Leray-Schauder degree theory to continue this solution along the homotopy. The challenge here lies in the fact that
there we always have a trivial periodic solution, given by the disease-free equilibrium, which lies on the boundary of the relevant domain $D$
in the functional space, which requires us to construct a smaller domain $U\subset D$ excluding the trivial solution, and to show that the conditions
for applying the Leray-Schauder theory hold for the domain $U$.

We note that our proof of Theorem \ref{main} is easily extended to give the same result for the SIRS model, which includes loss of immunity
\begin{equation}\label{S5}S'=\alpha S+\mu(1-S) -\beta(t) SI,\end{equation}
\begin{equation}\label{I5}I'=\beta(t) SI-(\gamma+\mu) I,\end{equation}
\begin{equation}\label{R5}R'=\gamma I -(\mu+\alpha) R.\end{equation}
We present the proof for the SIR model ($\alpha=0$) in order to avoid notational clutter.

In section \ref{proof} we prove Theorem \ref{main}. In section \ref{numerical} we discuss how to obtain the $T$-periodic solution numerically, which cannot
be done by direct numerical simulation when it is unstable, and present some results obtained by using the Galerkin method.
Finally, in section \ref{discussion} we mention some other works providing rigorous mathematical results on forced SIR models, beyond
numerical simulation.

\section{Proof of the Theorem}
\label{proof}
Since $S(t),I(t),R(t)$ are fractions of the population we have $S(t)+I(t)+R(t)=1$ for all $t$ - note that by adding the equations (\ref{S})-(\ref{R}) we have
$(S(t)+I(t)+R(t))'=0$. Since $R$ does not appear in (\ref{S}),(\ref{I}), the equation (\ref{R}) can be ignored, and it suffices to proved the existence of a periodic solution of (\ref{S}),(\ref{I}) satisfying
\begin{equation}\label{pos}S(t)>0,\;I(t)>0,\;S(t)+I(t)<1,\;\;\forall t,\end{equation}
where the third condition is equivalent to $R(t)=1-I(t)-S(t)>0$.

We decompose $\beta(t)$ as
$$\beta(t)=\bar{\beta}+\beta_0(t),\;\;\int_0^T \beta_0(t)dt=0.$$
Setting, for $\lambda\in [0,1]$,
\begin{equation}\label{deco}\beta_{\lambda}(t)=\bar{\beta}+\lambda\beta_0(t),\end{equation}
we consider the system
\begin{equation}\label{S1}S'=\mu(1-S) -\beta_{\lambda}(t) SI\end{equation}
\begin{equation}\label{I1}I'=\beta_{\lambda}(t) SI-(\gamma+\mu) I,\end{equation}
which is homotopy between an unforced system with $\beta_0(t)=\bar{\beta}$ and our system (\ref{S}),(\ref{I}), which corresponds to $\lambda=1$.
For $\lambda=0$, (\ref{S1}),(\ref{I1}) has exactly two periodic solutions, which are constant, given by
\begin{equation}\label{I00}S_0=1,\;\;I_0=0,\;\;\end{equation}
\begin{equation}\label{I0}S^*=\frac{\gamma+\mu}{\bar{\beta}},\;\;I^*=\mu\Big(\frac{1}{\gamma+\mu}-\frac{1}{\bar{\beta}}\Big).
\end{equation}
We note that $(S_0,I_0)$ (the disease-free equilibrium) is in fact a (trivial) periodic solution of (\ref{S1})-(\ref{I1}) for {\it{all}} $\lambda$.
Our aim is to continue the solution $(S^*,I^*)$ with respect to $\lambda$ in order to prove the existence of a periodic solution
for $\lambda=1$. To this end, we now reformulate the problem in a functional-analytic setting, which will enable us to employ
degree theory.

We rewrite (\ref{S1}),(\ref{I1}) as
\begin{equation}\label{S2}S'+\mu S=\mu -\beta_{\lambda}(t) SI\end{equation}
\begin{equation}\label{I2}I'+(\gamma+\mu) I=\beta_{\lambda}(t) SI\end{equation}

Let $X,Y$ be the Banach spaces
$$X=\{ (S,I)\;|\; S,I \in C^1(\Real),\; S(t+T)=S(t), I(t+T)=I(t)\}$$
$$Y=\{ (S,I)\;|\; S,I \in C^0(\Real),\; S(t+T)=S(t), I(t+T)=I(t)\}$$

Define the linear operator $L:X\rightarrow Y$ by
$$L(S,I)=(S'+\mu S,I'+(\gamma+\mu) I)$$
and the nonlinear operator  $N:Y\rightarrow Y$
$$N_\lambda(S,I)=(\mu -\beta_{\lambda}(t) SI,\beta_{\lambda}(t) SI).$$
Then the periodic problem for (\ref{S2})-(\ref{I2}) can be rewritten as
\begin{equation}\label{re}L(S,I)=N_{\lambda}(S,I).\end{equation}
It is easy to check that $L$ is invertible, that is the equations $S'+\mu S=f$ and $I'+(\gamma+\mu)I=g$ have unique $C^1$ $T$-periodic solutions
$S,I$ for any $f,g\in Y$, and the mapping $L^{-1}:Y\rightarrow X$ given by $L^{-1}(f,g) =(S,I)$ is bounded.
We can thus rewrite (\ref{re}) as
\begin{equation}\label{re1}F_{\lambda}(S,I)=0.\end{equation}
where $F_{\lambda}:Y\rightarrow X$ is given by
\begin{equation}\label{fr}F_\lambda(S,I)=(S,I)-L^{-1}\circ N_\lambda(S,I).\end{equation}
Since $L^{-1}:Y\rightarrow X$ is bounded, and since, by the Arzela-Ascoli Theorem, $X$ is compactly embedded in $Y$, we can consider $L^{-1}$ as a compact operator
from $Y$ to itself, and since $N:Y\rightarrow Y$ is continuous, $L^{-1}\circ N_\lambda$ is compact as an operator from $Y$ to itself. We therefore
consider (\ref{re1}) in the space $Y$, and we note that any solution in $Y$ will in fact be in $X$, hence a classical solution of (\ref{S2}),(\ref{I2}).
Since $F_{\lambda}$ is a compact perturbation of the identity on $Y$, Leray-Schauder theory is applicable.
Since we want our solution to satisfy (\ref{pos}), we want to solve (\ref{re1}) in the subset $D\subset Y$ given by
$$D=\{ (S,I)\in Y\;|\; S(t)>0,\; I(t)>0,\; S(t)+I(t)<1\}.$$
Note that for $\lambda=0$ the solution $(S^*,I^*)$ given by (\ref{I0}) lies in $D$. Our aim is to continue this solution in $\lambda$ up to
$\lambda=1$.

We recall that the Leray-Schauder degree theory (see e.g. \cite{brown,zeidler}) implies that, given a bounded open set $U\subset Y$, the existence of a solution $(S,I)$ of
(\ref{re1}) for all $\lambda\in [0,1]$ will be assured if the following conditions hold:

(I) $(S^*,I^*)\in U$,

(II) $deg(F_0,(S^*,I^*))\neq 0$,

(III) $F_\lambda(S,I)\neq 0$ for all $(S,I)\in \partial U$, $\lambda\in [0,1]$.

The most obvious choice for $U$ would be $U=D$. However, this will not do, since $(S_0,I_0)$ (given by (\ref{I00})) satisfies $(S_0,I_0)\in \partial D$ and $F_\lambda(S_0,I_0)=0$,
so that (III) does not hold.
To satisfy (III) we will need to choose $U$ so as to exclude $(S_0,I_0)$ from its boundary.
We take $U$ to be the open subset of $D$ given by
\begin{equation}\label{U}U=\{ (S,I)\in D\;|\; \min_{t\in\Real}S(t)<\delta\},\end{equation}
where $\delta\in (0,1)$ is fixed. Note that $(S_0,I_0)\not\in \bar{U}$.
We will show below that $U$ satisfies (I)-(III) if $\delta$ is chosen so that $\delta\in ({\cal{R}}_0^{-1},1)$.

We first show that $(S_0,I_0)$ is the {\it{only}} solution of (\ref{re1}) on $\partial D$.

\begin{lemma}\label{only} If $(S,I)\in\partial D$ is a solution of (\ref{re1}) for some $\lambda\in [0,1]$, then $(S,I)=(S_0,I_0)$, as given by (\ref{I00}).
\end{lemma}

\begin{proof} Assume that $(S,I)\in \partial D$ is a solution of (\ref{S2}),(\ref{I2}).
Note that $(S,I)\in \partial D$, if an only if
\begin{equation}\label{cl}S(t)\geq 0,\; I(t)\geq 0,\; S(t)+I(t)\leq 1,\end{equation}
and at least one of the following conditions holds:

(i) There exists $t_0\in \Real$ so that $I(t_0)=0$.

(ii) There exists $t_0\in \Real$ so that $S(t_0)=0$.

(iii) There exists $t_0\in \Real$ so that $S(t_0)+I(t_0)=1$.

We now consider each of these three cases:

(1)
Assume (i) holds. Let $\tilde{S}$ be the solution of
$$\tilde{S}'=\mu(1-\tilde{S}),\;\; \tilde{S}(t_0)=S(t_0).$$
and let $\tilde{I}(t)=0$. Then $\tilde{S},\tilde{I}$ is
 a solution of the initial-value problem (\ref{S2}),(\ref{I2}) with initial
condition
$$\tilde{S}(t_0)=S(t_0),\;\tilde{I}(t_0)=0.$$
By uniqueness of the solution for the initial-value problem, we conclude that $S=\tilde{S}$, $I=0$. Thus $S$ satisfies
$S'=\mu(1-S)$,
and since the only periodic solution of this equation is $S=1$, we conclude that $(S,I)=(1,0)$, as we wanted to prove.

(2) Assume now that (ii) holds.
Then from (\ref{S2}) we get $S'(t_0)=\mu$. But this implies that $S(t)<0$ for $t<t_0$ sufficiently close to $t_0$, which
contradicts (\ref{cl}). Thus this case is impossible.

(3) Assume now that (iii) holds. Moreover since we have already proven the result in the case that (i) holds, we may assume that $I(t)>0$ for
all $t$. Adding (\ref{S1}) and (\ref{I1}) we get
$$(S+I)'(t_0)=\mu(1-S(t_0)-I(t_0))-\gamma I(t_0)=-\gamma I(t_0)<0.$$
Therefore we conclude that $S(t)+I(t)>1$ for $t<t_0$ sufficiently close to $t_0$, contradicting (\ref{cl}). Therefore this case is impossible.
\end{proof}

We can now show that $U$, defined by (\ref{U}), satisfies (III).

\begin{lemma}Assume ${\cal{R}}_0>1$. If $\frac{1}{{\cal{R}}_0}<\delta<1$ then, for any $\lambda\in [0,1]$ there are no solutions $(S,I)$ of (\ref{re1}) with
$(S,I)\in\partial U$.
\end{lemma}

\begin{proof} Suppose $(S,I)\in\partial U$. Then either $(S,I)\in \partial D$ or $(S,I)\in D$ and
\begin{equation}\label{bc}
\min_{t\in\Real}S(t)=\delta.
\end{equation}
In the first case, Lemma \ref{only} and the fact that $(S_0,I_0)\not\in\partial U$ imply that $(S,I)$ is not a solution of (\ref{re1}). We therefore assume that
 $(S,I)\in D$ and (\ref{bc}) holds, which implies
that
\begin{equation}\label{bc1}
S(t)\geq \delta, \;\;\forall t.
\end{equation}
Assume by way of contradiction that $(S,I)$ solves (\ref{re1}), or equivalently $(S,I)$  solves (\ref{S2}),(\ref{I2}). Using the assumption $(S,I)\in D$, we have that $I$ is everywhere positive, so we can divide (\ref{I2}) by $I$, and
integrate over $[0,T]$, to obtain
\begin{equation}\label{ee}\frac{1}{T}\int_0^T \beta_\lambda(t)S(t)dt=\gamma+\mu.\end{equation}
But from (\ref{bc1})  we get
\begin{equation}\label{vv}\frac{1}{T}\int_0^T \beta_\lambda(t)S(t)dt \geq \delta \frac{1}{T}\int_0^T \beta(t)dt=\delta \bar{\beta}.\end{equation}
By the assumption $\delta>\frac{1}{{\cal{R}}_0}$ we have $\delta \bar{\beta}>\gamma+\mu$, so that (\ref{vv}) implies
\begin{equation*}\label{vv1}\frac{1}{T}\int_0^T \beta_\lambda(t)S(t)dt >\gamma+\mu,\end{equation*}
contradicting (\ref{ee}).
\end{proof}

To apply the Leray-Schauder degree it remains to verify that (I) and (II) hold. Since $S^*=\frac{1}{{\cal{R}}_0}$, the condition  $\delta> \frac{1}{{\cal{R}}_0}$
implies $(S^*,I^*)\in U$, so (I) holds. To prove (II), it suffices to show that the Fr\'echet derivative $DF_{0}(S^*,I^*)$ is invertible.
Since $F$ is a compact perturbation of the identity so that $DF_{0}(S^*,I^*)$ is Fredholm, it suffices to prove that the kernel of $DF_{0}(S^*,I^*)$ is trivial. Indeed, let us assume that $(V,W)\in \ker(DF_{0}(S^*,I^*))$, and
prove that $(V,W)=0$. We have $DF_{0}(S^*,I^*)(V,W)=0$, or, equivalently,
\begin{equation}\label{ker}L(V,W)=DN_0(S^*,I^*)(V,W).\end{equation}
Note that
$$DN_0(S^*,I^*)(V,W)=(-\bar{\beta} (S^*W+I^*V),\bar{\beta}  (S^*W+I^*V)),$$
so that (\ref{ker}) is equivalent to
%$$V'+\mu V=-\bar{\beta}  (S^*W+I^*V),$$
%$$W'+(\gamma+\mu) W=\bar{\beta}  (S^*W+I^*V),$$
%that is
\begin{equation}\label{sys}\left(
    \begin{array}{c}
      V \\
      W \\
    \end{array}
  \right)'=\left(
             \begin{array}{cc}
               -\mu {\cal{R}}_0 & -(\gamma+\mu) \\
               \mu({\cal{R}}_0-1) & 0 \\
             \end{array}
           \right)\left(
                    \begin{array}{c}
                      V \\
                      W \\
                    \end{array}
                  \right).
\end{equation}
The characteristic polynomial of the above matrix is
$$p(x)=x^2+\mu {\cal{R}}_0 x+ (\gamma+\mu)\mu({\cal{R}}_0-1).$$
Noting that $p(0)>0$ and that, for $\omega\in \Real$, $Im(p(\omega i))=\mu {\cal{R}}_0 \omega$, we see that the matrix has no imaginary or $0$
eigenvalues, so that (\ref{sys}) has no periodic solutions except $(V,W)=(0,0)$, and the claim is proved.

We have thus proven that (I)-(III) hold, which completes the proof of Theorem \ref{main}.

\section{Observing the $T$-periodic solution}
\label{numerical}

As we have noted in the Introduction, the period solution whose existence was proved whenever ${\cal{R}}_0>1$ is observable in numerical simulation of (\ref{S})-(\ref{R})
only for those parameter regimes for which it is stable. Therefore, if we are interested in examining the shape and amplitude of the $T$-periodic solution for values of the parameters for which the system displays subharmonic or chaotic behavior, we need a different computational approach. We now describe a simple
approach that we successfully implemented, which allows us to observe the $T$-periodic solution for arbitrary parameters. We used the Galerkin
method, expanding the periodic solution $S(t),I(t)$ in a Fourier series. We used the Maple system, whose symbolic capabilities make the implementation
particularly easy.
 We take $T=2\pi$,
\begin{equation}\label{for}\beta(t)=\bar{\beta}(1+\lambda \cos(t)), \end{equation}
and search for
approximate periodic solutions of (\ref{S}),(\ref{I}) the form
\begin{eqnarray}\label{fo}\tilde{S}(t)&=&A_S^{0}+\sum_{n=1}^N [A_S^{n}\cos(nt)+B_S^{n}\sin(nt)],\nonumber\\
\tilde{I}(t)&=&A_I^{0}+\sum_{n=1}^N [A_I^{n}\cos(nt)+B_I^{n}\sin(nt)].
\end{eqnarray}
Plugging (\ref{fo}) into (\ref{S})-(\ref{I}), and then taking the Fourier coefficients of both sides of the equations with respect to
$$\{1,\cos(t),\cdots,\cos(Nt),\sin(t),\cdots,\sin(Nt)\},$$
we get $4N+2$ algebraic equations in $4N+2$ variables, which we numerically solve for $A_S^0,\cdots A_S^N,B_S^1,\cdots,B_S^N,A_I^0,\cdots A_I^N,B_I^1,\cdots,B_I^N$, obtaining the approximate $T$-periodic solution (\ref{fo}). We have found that the numerical iteration for solving the algebraic equations, using Maple's
fsolve command, works well, when we start with the initial conditions for the iteration given by the endemic equilibrium of the autonomous case, that
is $A_S^0=S^*,I_S^0=I^*$ (see (\ref{I0})) and $A_S^k=B_S^k=A_I^k=B_I^k=0$ for $1\leq k\leq N$. We check that the functions $\tilde{S}(t),\tilde{I}(t)$ indeed
approximate a periodic solution of (\ref{S}),(\ref{I}) by observing that the highest Fourier coefficients $A_S^N,B_S^N,A_I^N,B_I^N$ are very small, and
by plugging $\tilde{S}(t),\tilde{I}(t)$ into (\ref{S}),(\ref{I}) and checking that the residual is small. We note that theoretical justification
of the Galerkin method for approximating periodic solutions can be found, e.g., in \cite{bobylev}.

We now present some examples of results obtained by the method described above.
With the period $2\pi$ of the forcing representing one year, we took take $\gamma$ corresponding to a $2$-week infectious period, $\bar{\beta}=20\gamma$, $\mu$ corresponding to $\%4$ population growth rate
per year, giving ${\cal{R}}_0=19.97$. These parameters are approximately those estimated for measles. We consider different values of the
strength of seasonality $\lambda$ (see (\ref{for})). In figure 1 we plot, for different value of $\lambda$, the periodic solution found by the Galerkin
method (with $N=8$), together with a solution of (\ref{S})-(\ref{I}) obtained by direct simulation, starting the plot at $t=3000\pi$ to ensure that transients have decayed.

\begin{figure}\label{fig1}
\centering
   \includegraphics[height=4.5cm,width=4.5cm, angle=0]{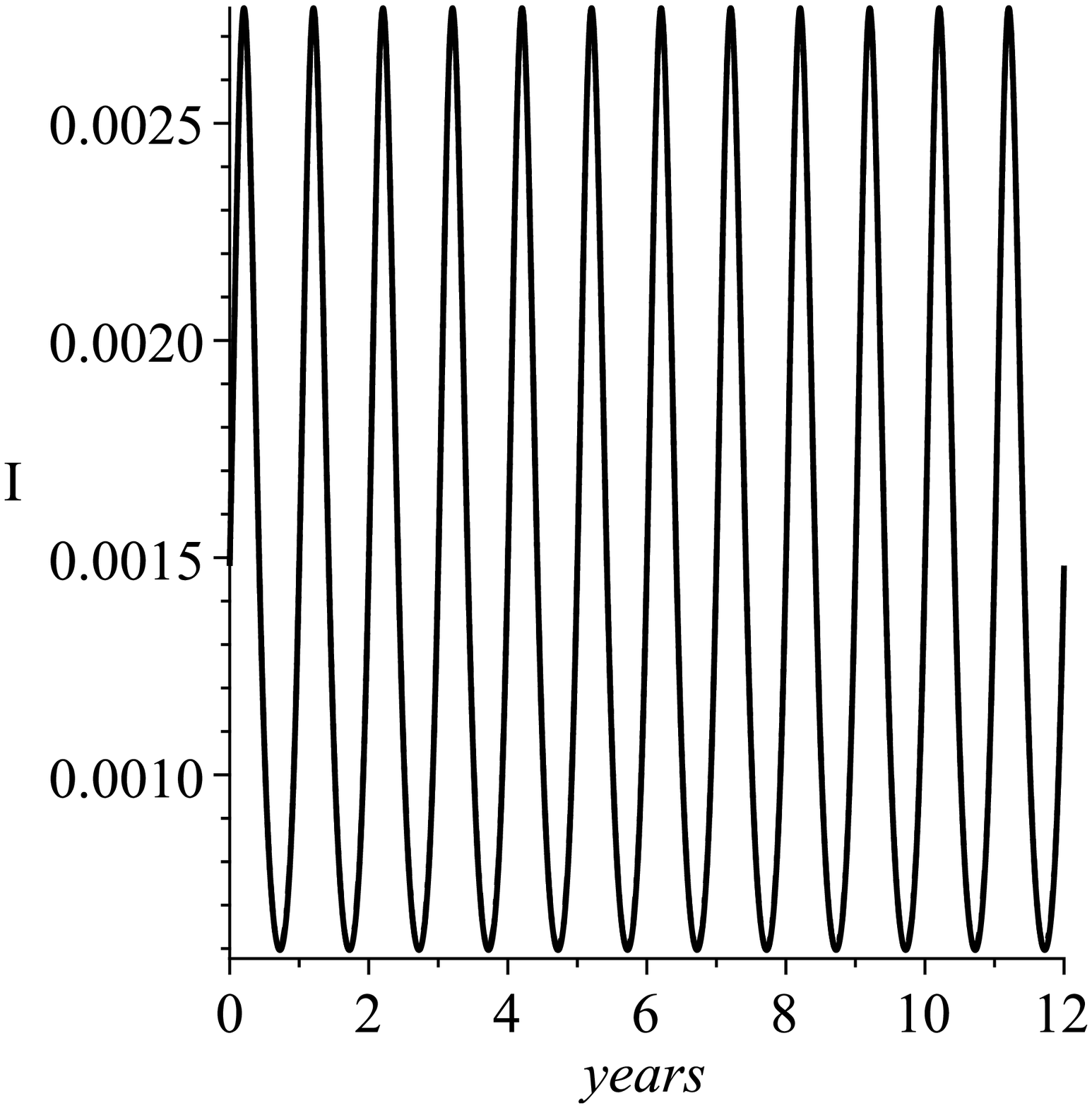}\includegraphics[height=4.5cm,width=4.5cm, angle=0]{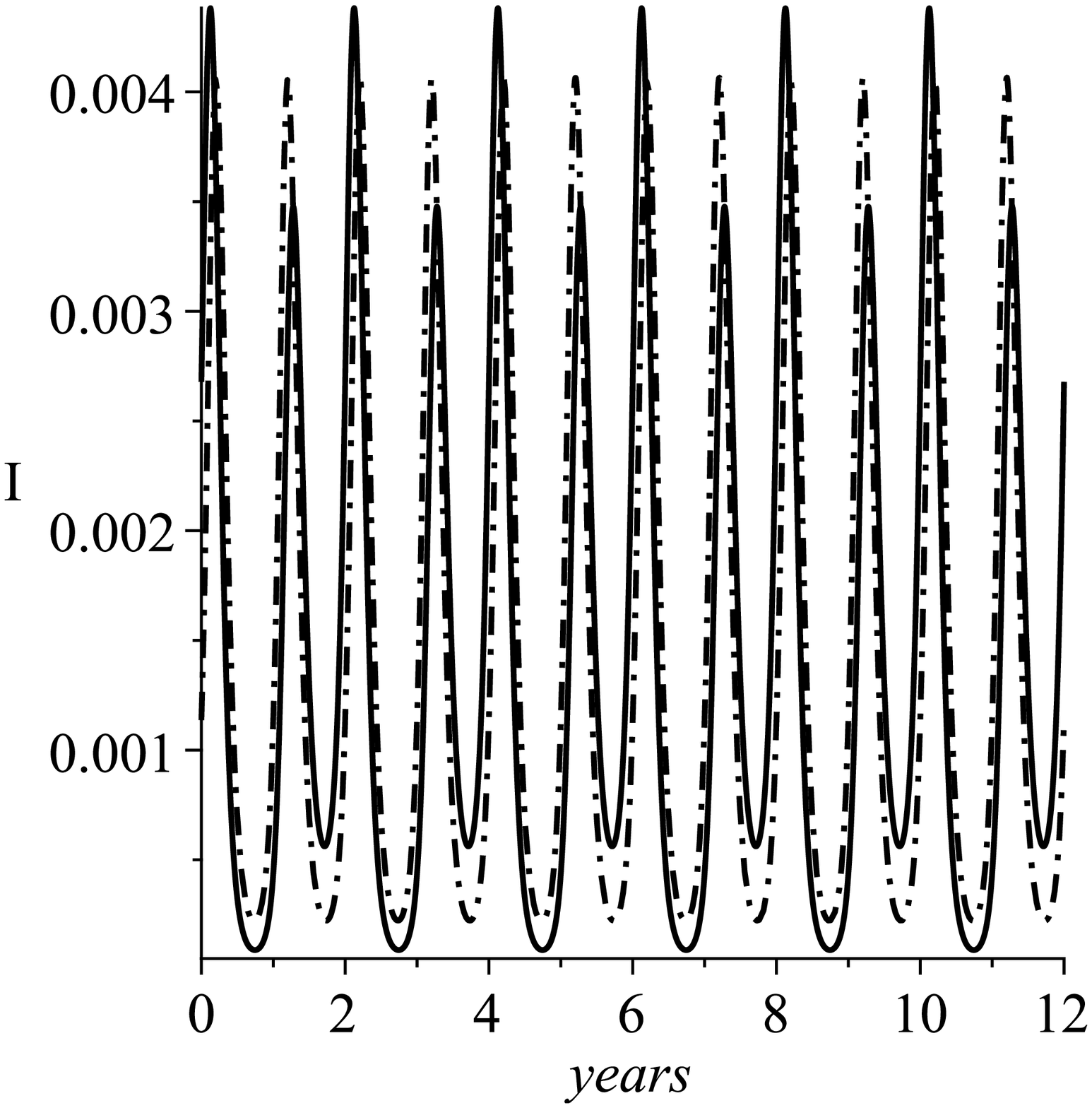}\includegraphics[height=4.5cm,width=4.5cm, angle=0]{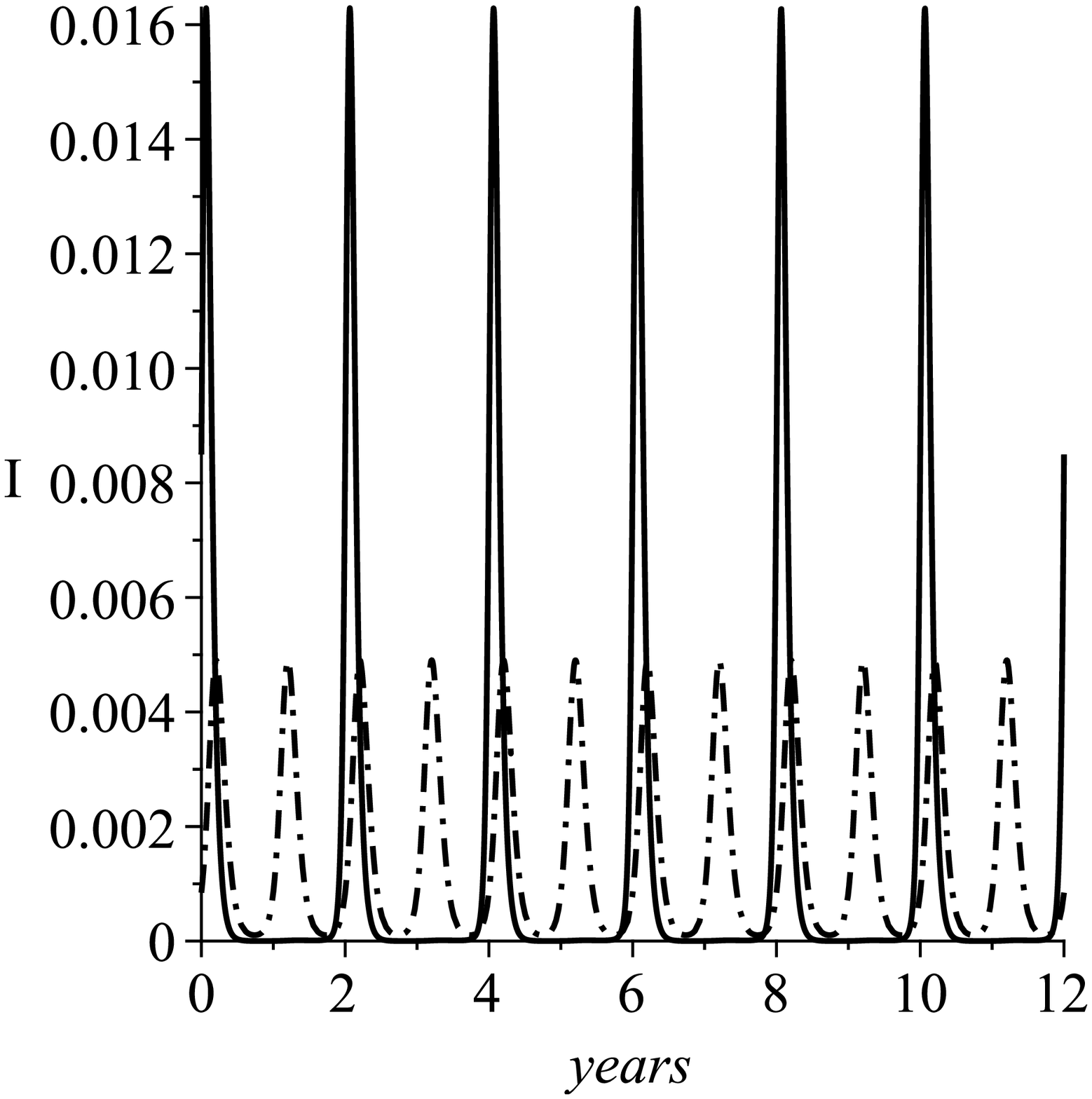}\\
   \includegraphics[height=4.5cm,width=4.5cm, angle=0]{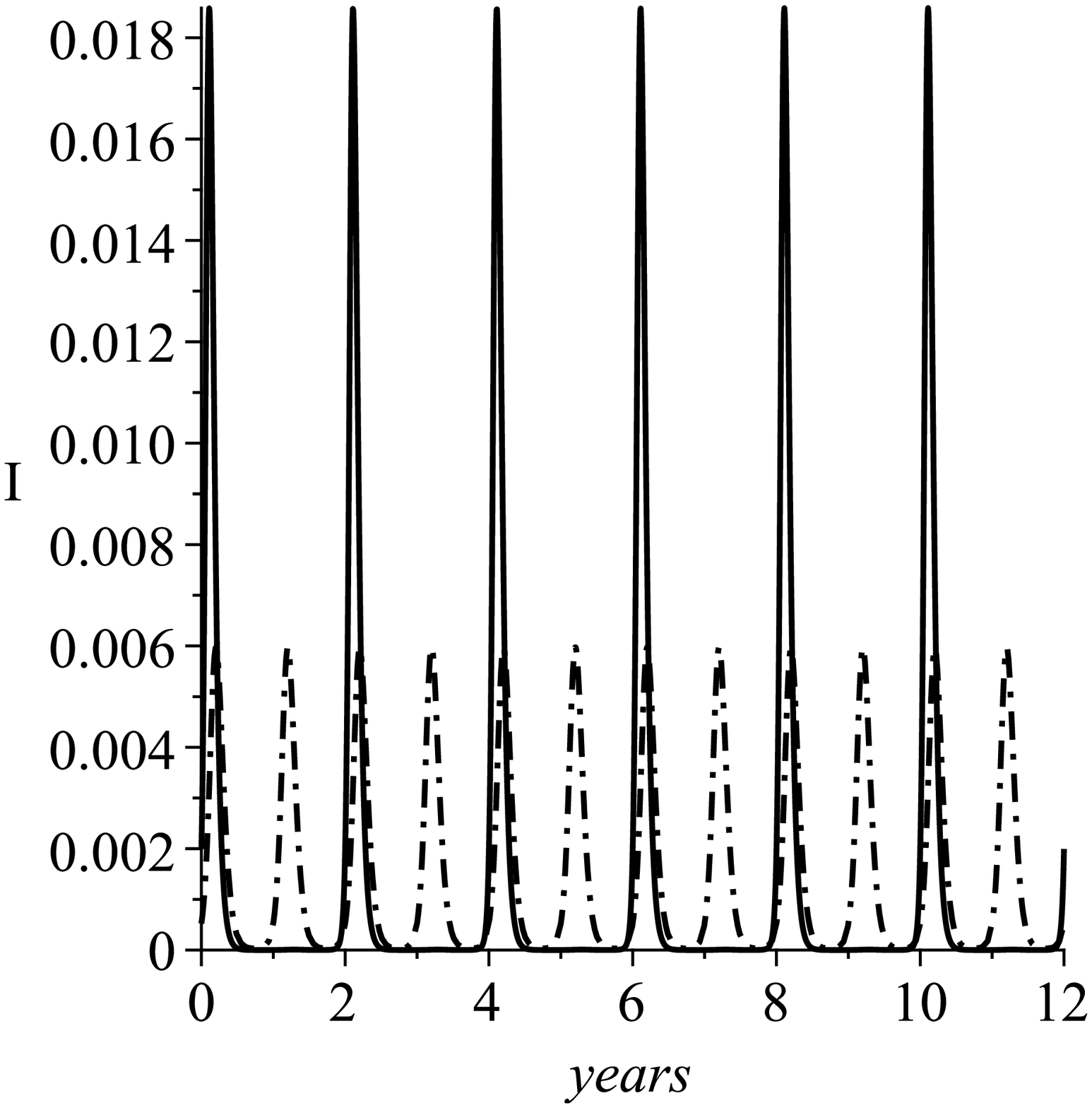}\includegraphics[height=4.5cm,width=4.5cm, angle=0]{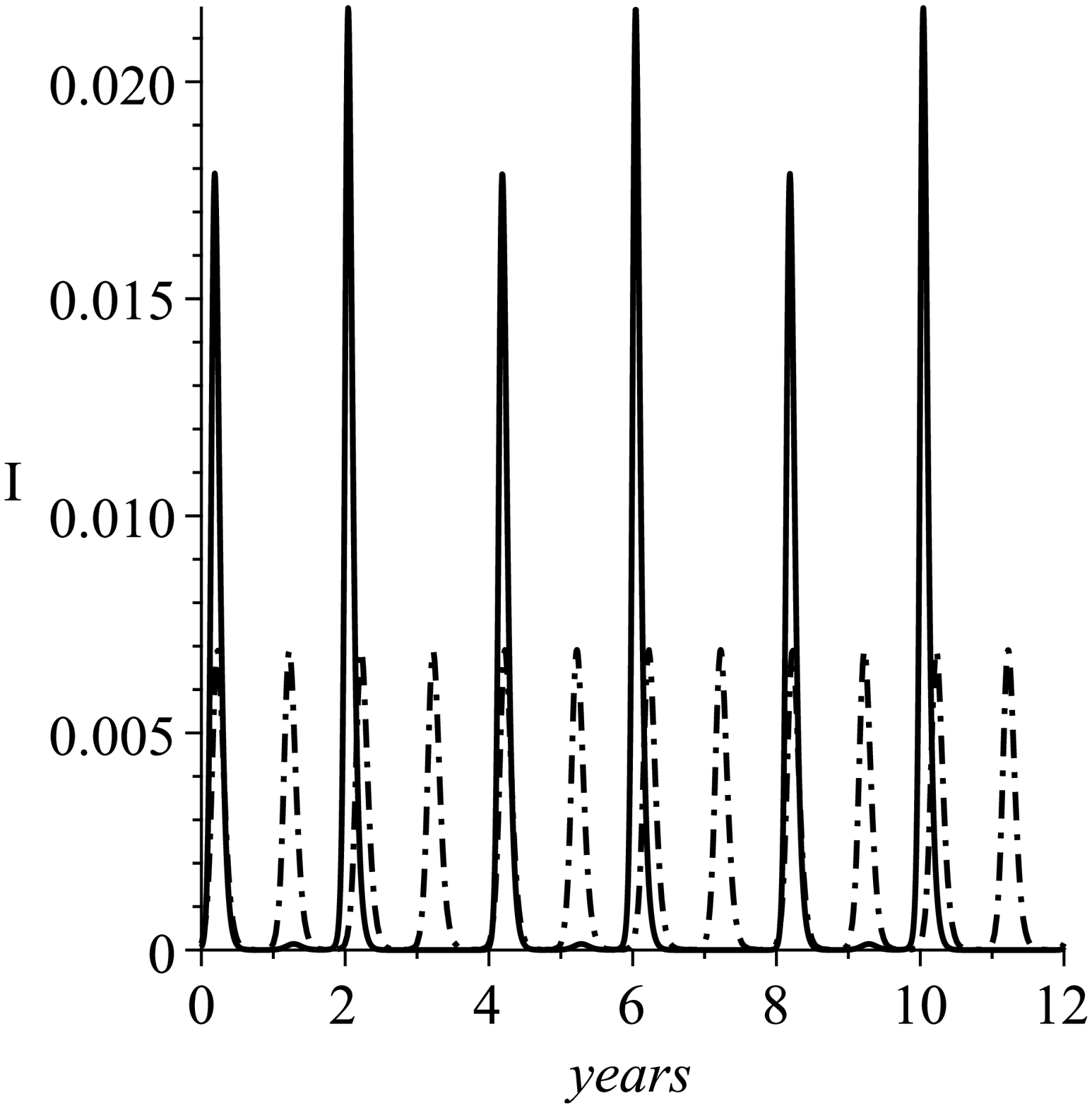}\includegraphics[height=4.5cm,width=4.5cm, angle=0]{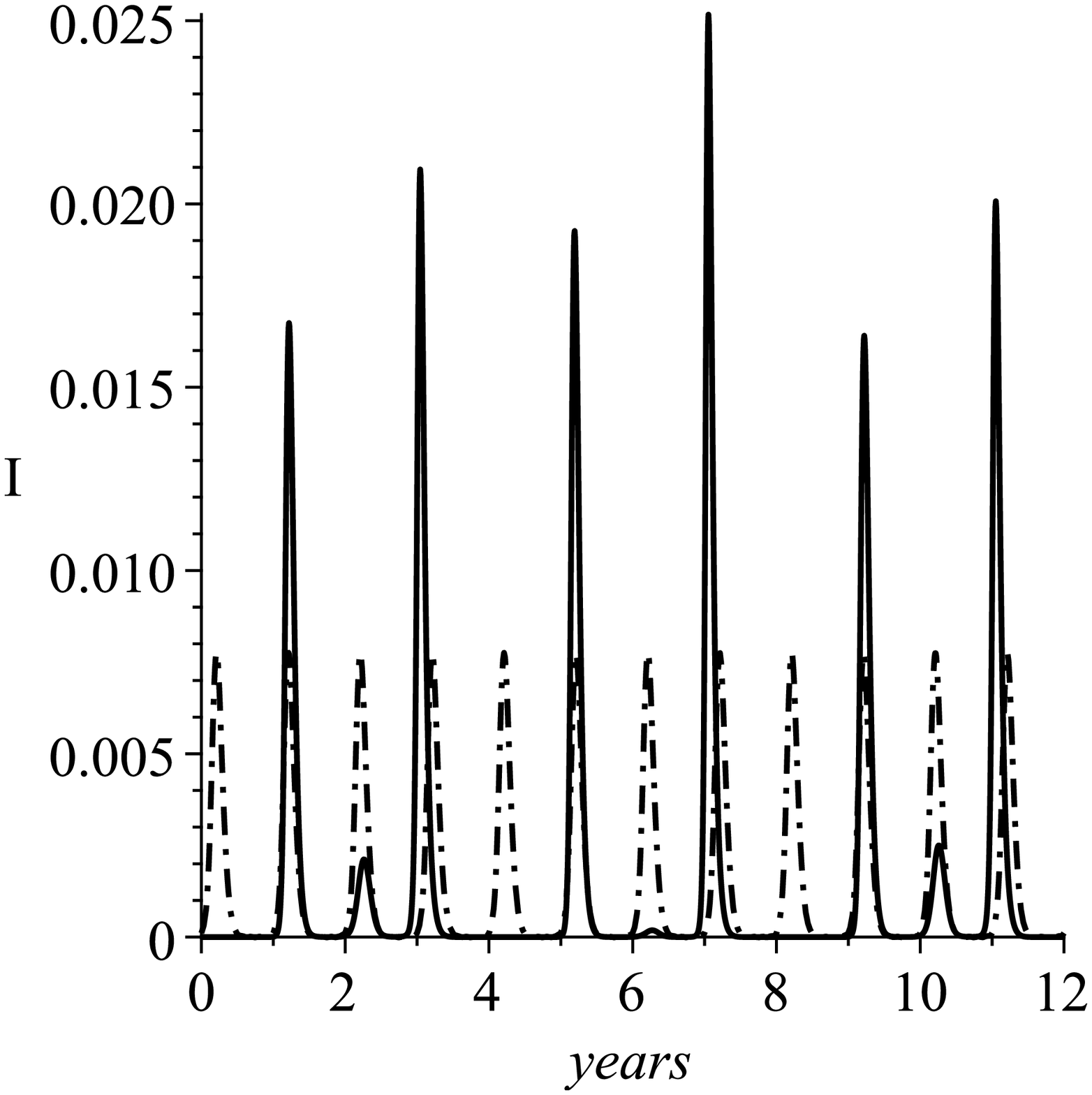}

    \caption{Solutions of the SIR model with $T$-periodic forcing ($T=2\pi$), obtained by direct simulation, and the $T$-periodic solution obtained by the Galerkin method (dashed line), for varying strength of seasonality $\lambda$. Top row, from left to right: $\lambda=0.1,0.21,0.3$, Bottom row: $\lambda=0.45,0.6,0.7$. Other parameters: $\gamma=14\frac{2\pi}{365}$,  $\bar{\beta}=20\gamma$, $\mu=\frac{0.04}{2\pi}$.}
\end{figure}

\begin{figure}\label{fig2}
\centering
   \includegraphics[height=5cm,width=10cm, angle=0]{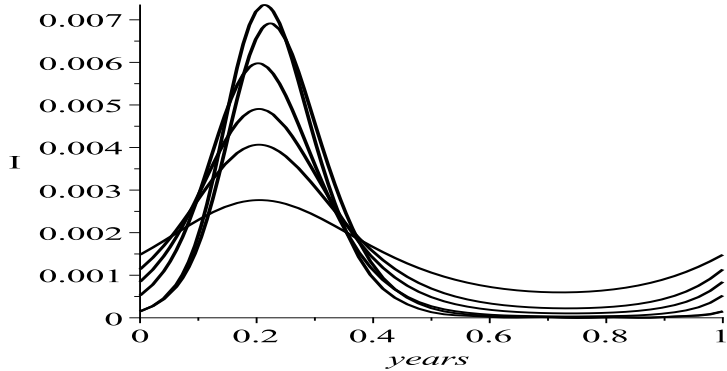}
\caption{The $T$-periodic solution obtained by the Galerkin method, for varying strength of seasonality $\lambda=0.1,0.21,0.3,0.45,0.6,0.7$.}
\end{figure}

When $\lambda=0.1$, the system behavior is $2\pi$-periodic, so the solution of the simulated system coincides with the $2\pi$-periodic solution found by
the Galerkin method. At $\lambda=0.21$, the $2\pi$-periodic solution has lost stability, and we see bifurcation to a subharmonic of order 2 (period $4\pi$), which is still quite close to the $2\pi$-periodic solution, with larger and smaller epidemics alternating. At $\lambda=0.3,0.45$ the $4\pi$-periodic subharmonic solution is already quite different, with a large epidemic every two years. At $\lambda=0.6$ we observe that the system has a subharmonic of order 4 (period $8\pi$), while at
$\lambda=0.7$ we observe chaotic behavior. The $2\pi$-periodic solution (which is unstable
except for the case $\lambda=0.1$) increases in amplitude and becomes less sinusoidal as $\lambda$ increases (note the differences in scales in the
different plots). In figure 2 we plot the $2\pi$-periodic solutions for all values of $\lambda$, for a better view.

\section{Discussion}
\label{discussion}

The forced SIR model is a beautiful example of a simple nonlinear dynamical system which displays complicated behaviors which are difficult to
understand in intuitive terms. Moreover, these complicated behaviors are relevant to explaining the epidemiology of infectious diseases in humans,
as studies comparing the behavior of the SIR and variants of it to surveillance data have shown \cite{bolker,earn,keeling1}. We have proven the fundamental result
that a $T$-periodic solution exists for the $T$-periodically forced SIR model whenever ${\cal{R}}_0>1$. As we have stressed, this
does not mean that the dynamics of the model is periodic, since the periodic solution whose existence is proved need not be stable, although one can use standard
 perturbation theory to prove that the $T$-periodic solution {\it{is}} stable provided the seasonality parameter $\lambda$ in (\ref{deco}) is sufficiently small.
Numerical simulations show that complex dynamics - subharmonic and chaotic behavior - is very common in the forced SIR model. It is interesting to ask
to what extent the complex dynamics of the forced SIR model can be rigorously understood, beyond numerical simulations. While we do not expect
to be able to precisely characterize the dynamics of the model for different parameter values, it is of great interest even to be able to rigourously
prove that complicated dynamics occurs for at least {\it{some}} parameter values. In this context we mention the work
of H.L. Smith \cite{smith1,smith2}, who proved that the forced SIR model can have multiple stable subharmonic oscillations in certain parameter
ranges. Chaotic behavior has been rigorously established by Glendinning \& Perry \cite{glendinning} for a variant of the forced SIR model, in which
the dependence of the incidence term on $I$ is nonlinear. For the standard SIR model (\ref{S})-(\ref{R}), we are not aware of a proof of chaotic behavior.
Classifying and explaining the dynamical patterns observed in simulations of the forced SIR model is still very challenging, so that, like other
well-known `simple' models such as the forced pendulum equation, the forced SIR model can serve as a stimulus and as a benchmark problem for new developments in
nonlinear analysis.

{\bf{Funding:}} The author acknowledges support of EU-FP7 grant Epiwork.

\end{document}